\documentclass[runningheads]{llncs}
\usepackage[T1]{fontenc}

\usepackage{amsmath,amssymb,mathtools}
\usepackage{microtype}
\usepackage{booktabs}
\usepackage{enumitem}
\usepackage{url}

\newcommand{\E}{\mathbb{E}}
\newcommand{\Prb}{\mathbb{P}}

\newcommand{\ceil}[1]{\left\lceil #1\right\rceil}
\newcommand{\KL}{D}
\newcommand{\id}{\mathsf{id}}
\newcommand{\ticket}{\mathsf{ticket}}
\newcommand{\addr}{\mathsf{addr}}
\newcommand{\hash}{\mathsf{H}}
\newcommand{\Sig}{\mathsf{Sig}}
\newcommand{\Team}{\textsf{TIMELY-TEAM}}
\newcommand{\Pivot}{\Team}
\newcommand{\Nstar}{N^\star}
\newcommand{\Perm}{\mathfrak S}
\newcommand{\Sedna}{\textsc{Sedna}}
\newcommand{\PivotK}{\textsf{PIVOT-}K}
\providecommand{\Team}{\textsf{TIMELY-TEAM}}
\providecommand{\Nstar}{N^\star}

\providecommand{\Game}{\Gamma}

\title{Slack and Budget Breaking in Threshold Team
Production}
\titlerunning{Slack and Budget Breaking in Threshold Team Production}

\author{Benjamin Marsh\inst{1, 2} \and Alejandro Ranchal-Pedrosa\inst{1}}
\authorrunning{B. Marsh and A. Ranchal-Pedrosa}
\institute{Sei Labs \and University of Portsmouth\\\email{\{ben|alejandro\}@seinetwork.io}}

\begin{document}
\maketitle

\begin{abstract}
A threshold system completes a public task only after $\kappa$ verifiable shares are publicly committed. If the honest schedule creates
\(
    \Nstar=\kappa+\Delta
\) share opportunities by deadline $t^\star$, then $\Delta$ shares are slack such that a coalition delays completion if and only if it withholds at least $\Delta+1$ shares. The incentive problem is therefore to price the cheapest sabotage set. Agents receive a direct fee $f$ per committed share. A delaying coalition may also obtain delay value at most $L$, and may earn additional fee revenue during recovery after the deadline. Let $R_1^+$ be a pathwise upper bound on the coalition's incremental fee revenue in a recovery slot that completes the task, including any same-slot overshoot. The principal can post a nonnegative completion bounty that depends only on committed shares, uses no deposits or punishments, and expires if completion is late. The optimal rule is uniform, as if completion occurs by $t^\star$, every admissible horizon share receives $B/\Nstar$, otherwise no bounty is paid. Full participation is ex-post strongly delay proof exactly when
\(
    (\Delta+1)f+\frac{\Delta+1}{\Nstar}B
    \ge
    L+R_1^+ .
\) Equivalently, the exact worst-case budget is
\(
    B^\star
    =
    \frac{\Nstar}{\Delta+1}
    \bigl(L+R_1^+-(\Delta+1)f\bigr)^+ .
\) The bound is tight for every nonnegative completion-measurable bounty, among the $\Nstar$ horizon shares, some $\Delta+1$ receive total bounty at most $(\Delta+1)B/\Nstar$, and withholding precisely those shares delays completion. The result applies to threshold signatures, data availability certification, coded dissemination, and generic $k$-of-$n$ completion tasks. We also isolate a separate limit, no transfer rule based only on completed shares can remove a final slot race in which a coalition has already observed enough pre-completion shares to act.

\end{abstract}

\clearpage

\section{Introduction}
Many distributed systems complete a public task only after collecting a threshold number of verifiable shares. A certificate forms after enough validators sign. A data availability layer accepts a block after enough samples or attestations appear. A coded dissemination protocol reveals a payload after enough coded shares are publicly committed. A generic $k$-of-$n$ protocol completes once enough votes or certificate shares are included~\cite{luby2002,shokrollahi2006,rabin1989,cachin2005avid,sedna2025,marsh2026mechanismdesignoverviewsedna}. In all of these cases, completion is not produced by one agent but by a threshold team. This creates a simple sabotage problem as the principal values timely completion, while individual agents receive local fees for their own shares. A coalition may prefer to withhold enough of its shares to push completion past the honest deadline. In blockchain applications, one source of such delay value is maximal extractable value (MEV): profit from transaction ordering, front-running, liquidations, arbitrage, or related use of pre-confirmation information~\cite{eskandari2019,daian2020,qin2021flashloans}. In the model we do not need the source of the value.  We only require an upper bound $L$ on the coalition's redistributable gain from delay. 
Modelling the game as running in a series of slots that accumulate shares, let $t^\star$ be the honest horizon, the first slot by which the threshold would be reached if all available agents submitted their shares. If the honest schedule creates
\(
    \Nstar=\kappa+\Delta
\) share opportunities by this horizon, then $\Delta$ shares are redundant. If a coalition withholds $W$ horizon shares, completion is delayed if and only if
\(
    W>\Delta .
\) Thus zero slack is a knife edge such that one missing share delays completion. Positive slack changes the strategic problem. The first $\Delta$ withholdings do not delay anything unless they are coordinated with at least one more withholding. Slack converts unilateral free riding into coordinated sabotage.

This is a mechanism design problem with an external budget. In Holmstr\"om's teams model, efficient effort is obstructed by budget balance because agents do not internalize the value of team output~\cite{holmstrom1982}. Here the principal is an outside beneficiary and can break budget balance by posting a completion bounty. The relevant stability notion is coalitional: a delay attack need not be profitable for each deviator separately, but it may be profitable for a coalition that can redistribute the delay value among its members. We therefore use an ex-post strong-equilibrium style condition, in the lineage of strong equilibrium and coalition-proof Nash equilibrium~\cite{aumann1959,bernheim1987}. The robust design goal is Wilsonian: the contract should not depend on estimating cartel membership, network paths, or transaction-specific private values~\cite{wilson1987,bergemann2005}. 

Our main result characterizes the exact budget needed to deter coalitional delay. We consider nonnegative completion bounties with budget $B$. The bounty may depend on committed shares, the commitment schedule, and share admissibility, but it cannot use deposits, negative transfers, unverifiable cartel membership, or transaction specific private types. The optimal rule is the uniform timely team bounty, denoted \Team: if completion occurs by $t^\star$, every admissible horizon share receives $B/\Nstar$, if completion is late, the bounty expires and the escrow is returned. With direct fee $f$, delay value bounded by $L$, and recovery slot (slot exceeding the honest horizon for completion) fee leakage bounded by $R_1^+$, full participation is ex-post strongly delay proof exactly when \( (\Delta+1)f+\frac{\Delta+1}{\Nstar}B \;\ge\; L+R_1^+ . \) Equivalently,
\(
    B
    \ge
    \frac{\Nstar}{\Delta+1}
    \bigl(L+R_1^+-(\Delta+1)f\bigr)^+ .
\) The paper makes four contributions. First, it formalizes threshold team production with slack and proves the pathwise identity that delay occurs exactly when withholding exceeds slack. Second, it gives the exact minimax completion bounty among all nonnegative completion measurable transfer rules. Third, it shows how recovery fees, sender pricing, slack uncertainty, and exclusion ratchets affect the required budget. Fourth, it proves a separate final slot visibility impossibility, if a coalition has already observed enough pre-completion shares to act before completion, no transfer rule depending only on committed shares can eliminate that physical race.

\section{Model}
\label{sec:model}

The primitive object is a threshold public good, where a \emph{share} is any verifiable contribution that moves the system one unit closer to completion. A signature share, a data availability attestation, a coded data share, a vote, or another admissible proof. The project completes once at least $\kappa$ admissible shares have been publicly committed. There are $n$ agents and an outside beneficiary. Time is slotted. In slot $t$, an availability process selects a set $\mathcal A_t$ of agents who can each submit one admissible share. Write
\(
    a_t=|\mathcal A_t|,
    \qquad
    M_t=\sum_{u\le t}a_u .
\)
If every available agent submits, the first slot in which the threshold can be reached is
\(
    t^\star=\inf\{t:M_t\ge \kappa\}.
\)
We call $t^\star$ the honest horizon.  We call the availability process \emph{exchangeable} if the law of $(\mathcal{A}_t)_t$
is invariant under any relabeling of agents. The homogeneous sampling model is exchangeable. The number of admissible share opportunities available by that horizon is
\(
    \Nstar=M_{t^\star}.
\)
The slack is
\(
    \Delta=\Nstar-\kappa .
\)
Thus the honest path contains $\Delta$ redundant share opportunities by the deadline. A coalition deviates by switching some of its available share opportunities from submission to withholding. Agents outside the tested coalition continue to submit. If the coalition withholds $W_{t^\star}$ shares by the honest horizon, then the number of committed shares at $t^\star$ is
\(
    \Nstar-W_{t^\star}
    =
    \kappa+\Delta-W_{t^\star}.
\)
Hence completion is delayed past $t^\star$ if and only if
\(
    W_{t^\star}>\Delta .
\) This identity is the basic geometry of the model. Zero slack makes delay unilateral, positive slack means the first $\Delta$ withholdings do not delay anything unless they are coordinated with at least one more withholding. In the homogeneous case used for probability calculations below, $a_t=m$ for every slot. Then
\(
    t^\star=\ceil{\kappa/m},
    \qquad
    \Nstar=t^\star m,
    \qquad
    \Delta=t^\star m-\kappa .
\) The same notation applies across several systems. Threshold signatures where shares are signature shares, data availability where shares are attestations or coded data shares, coded dissemination where shares are agent bound coded fragments, and $k$-of-$n$ completion where shares are votes or certificate fragments. The incentive problem is the same in all cases: the principal needs $\kappa$ committed shares by $t^\star$, while a coalition may prefer to delete enough of its own share opportunities to make completion late.

\paragraph{Coded dissemination as an instantiation.}
In coded dissemination, a share is an admissible agent-bound coded fragment. A transaction completes once $\kappa$ such fragments have been publicly committed. The coding layer determines how fragments map to symbols and how many symbols are sufficient for reconstruction, but these details are not primitives of the incentive theorem. What matters for the model is that each admissible fragment is verifiable, bound to a particular agent or validity window, and counts as one share toward the threshold. Appendix~\ref{app:sedna-specifics} gives a concrete instantiation using the \Sedna-specific ticket format, symbol-resolution convention, and rounding details.

\paragraph{Admissibility.}
We assume admissibility is publicly verifiable. A submitted share counts for completion, fees, and bounty payments only if it satisfies the committed validity rule. In coded dissemination this is enforced by agent-bound tickets: copied fragments or fragments submitted by the wrong agent do not count. This prevents reward sniping and lets the bounty be defined over committed shares. In the main model we take tickets to be \emph{slot bound}, a ticket is admissible only in its issued slot and expires thereafter. Slot binding is what makes the fee genuinely forfeited rather than deferred, so a withheld share cannot be redeemed for $f$ in a recovery slot. Strict slot binding carries a liveness cost in heterogeneous latency populations, an honest agent that misses its issued slot to a transient delay forfeits $f$ even though it could still contribute before the horizon. The safe relaxation is a fee ticket grace window that closes at $t^\star$, under which a slot $t$ ticket is redeemable in any slot up to $t^\star$ but never in a recovery slot. This leaves Lemma~\ref{lem:loss} intact, a coalition that withholds through $t^\star$ to force delay still forfeits $f$, since redeeming inside the window would land the share by $t^\star$ and complete on time, whereas a window reaching into recovery would let a withholder recapture $f$ and reopen the leak measured by $R_1^+$. Adopting the timely fee prefix of Section~\ref{sec:pricing-dynamics} as the default fee rule achieves the same liveness effect while also setting $R_1^+=0$. Each admissible submitted share pays direct revenue $f>0$. We assume static nonnegative marginal fees in the main model, including an additional admissible pre-horizon share cannot reduce the agent's other revenue before considering delay value. Appendix~\ref{app:negative-fees} states the weaker net marginal condition needed in fee markets with displacement costs.

\paragraph{Agents, coalitions, and contacts.}
All agents are strategic. Their on-path action is to submit an admissible share when selected, a deviation is a switch from submission to withholding for some selected opportunities. For probability statements, fix a coalition $C$ controlling $\beta n$ agents and call it the cartel. In the homogeneous sampling model, let $A_t$ be the number of agents in $C$ selected in slot $t$,
\(
    A_t\sim \mathrm{Hypergeom}(n,\beta n,m),
    \qquad
    S_t=\sum_{u\le t} A_u .
\) Sampling is without replacement within a slot and with replacement across slots, so the variables $A_t$ are independent across slots. A coalition policy $\pi$ chooses, for each slot $t$, a number
\(
    X_t^\pi\in\{0,\ldots,A_t\}
\) of its selected shares to submit. In the ex-post stability test below, agents outside the deviating coalition continue to submit. They are not assumed mechanically obedient, they are simply not members of the tested deviation. The coalition's withheld count and the public completed count by time $t$ are
\(
    W_t^\pi=S_t-\sum_{u\le t}X_u^\pi,
    \qquad
    U_t^\pi=tm-W_t^\pi .
\) The public completion time is
\(
    T^\pi=\inf\{t:U_t^\pi\ge\kappa\}.
\) Under full submission, $T=t^\star=\ceil{\kappa/m}$ deterministically and $\Delta=t^\star m-\kappa$.

\paragraph{Payoffs, game, and solution concept.}
The extensive form $\Game(B)$ is as follows. First, the sender commits to the availability schedule, admissibility rule, fee tickets, and bounty rule. Second, selected agents observe their own shares; members of a coalition observe their collective selections and can make side payments among themselves. Third, in each slot, every selected agent chooses submission or withholding. Committed shares determine whether and when the threshold is reached, and transfers are paid according to the committed rule. Finally, if the coalition has created a delay, it may extract value from the resulting information lead or timing advantage. Let $L$ be a transaction specific upper bound on the redistributable value available to any coalition that delays completion past $t^\star$. Direct revenues are measured in the same units, common discount factors can be folded into $f$, $B$, and the recovery term. Delay can also change fee revenue. After $t^\star$, the sender may continue routing shares until the threshold is reached. Some of those recovery shares may be submitted by the delaying coalition, creating fee revenue that is absent from the on-time baseline. We account for this separately from $L$. Consider a deviation that withholds $c$ pre-horizon shares. Since the honest path has slack $\Delta$, the post-horizon deficit is
\(
    d=c-\Delta>0 .
\) Let $R_{\rm rec}(c)$ denote the coalition's maximum incremental fee revenue after $t^\star$ when the sender continues until completion. Before the final recovery slot, fewer than $d$ recovery shares can have been committed; otherwise the task would already have completed. We therefore isolate the only potentially large term. The recovery slot in which completion occurs. Let $R_1^+$ be a pathwise upper bound on the coalition's incremental fee revenue in that slot, including any same-slot shares included after the threshold crossing share. In the default fee-leak model, $R_1^+\le mf$ pathwise and $\E[R_1^+]\approx \beta mf$ under fresh uniform sampling. If recovery fees are burned, returned to the sender, or otherwise not paid to post-horizon submitters, then $R_1^+=0$. The timely completion bounty below pays no bounty after $t^\star$, so recovery accounting concerns fees only.

\begin{definition}[Ex-post strong delay proofness]
Fix a realized contact path up to $t^\star$ and a recovery fee model. The full submission strategy profile is \emph{ex-post strongly delay proof} if there is no coalition of agents that can switch some of its selected opportunities from submission to withholding, obtain $T>t^\star$, and weakly increase every coalition member's payoff with at least one strict increase under arbitrary internal side payments funded by at most $L+R_{\rm rec}(c)$, where $c$ is the number of pre-horizon withholdings in the deviation.
\end{definition}
This is a strong equilibrium style, pathwise core condition for the submission outcome of $\Game(B)$. The probabilistic parameter $\beta$ matters for pricing and attack likelihood; the main incentive theorem below is conditional on the realized path and applies to any coalition that has enough selected opportunities to attempt delay. The condition presumes transferable utility: the delay value $L$ can be redistributed frictionlessly among coalition members, which is exactly what reduces coalitional stability to a comparison of aggregate opportunity cost against $L+R_1^+$. Any friction in the side payment technology only shrinks the set of jointly profitable deviations, so $B^\star$ is an upper bound on the budget required under frictional redistribution.

\section{Slack Geometry of Withholding}
\label{sec:geometry}

The key accounting identity is independent of the cartel's dynamic policy.

\begin{theorem}[Delay is exactly withholding beyond slack]
\label{thm:delay}
For every contact path and every coalition policy $\pi$,
\(
    T^\pi>t^\star
    \quad\Longleftrightarrow\quad
    W_{t^\star}^\pi>\Delta .
\) Consequently, any delay-inducing policy must withhold at least $\Delta+1$ shares before $t^\star$, and this number is sufficient whenever the coalition controls at least $\Delta+1$ selected horizon shares.
\end{theorem}

\begin{proof}
At the honest horizon, the full submission path contains
\(
    \Nstar=\kappa+\Delta
\) admissible horizon shares. Under policy $\pi$, exactly $W_{t^\star}^\pi$ of these shares are missing, so
\(
    U_{t^\star}^\pi
    =
    \kappa+\Delta-W_{t^\star}^\pi .
\) The threshold has not been reached by $t^\star$ if and only if this quantity is strictly below $\kappa$, equivalently $W_{t^\star}^\pi>\Delta$.
\end{proof}

\begin{corollary}[Knife edge]
If $\Delta=0$, then one withheld cartel share before $t^\star$ delays completion. Under full withholding, the delay probability is
\(
    q_0=\Prb[S_{t^\star}>0]
    =1-\left(\frac{\binom{(1-\beta)n}{m}}{\binom{n}{m}}\right)^{t^\star} .
\)
\end{corollary}

\begin{lemma}[Full withholding maximizes delay probability]
\label{lem:full-withholding}
For every dynamic policy $\pi$, $\Prb[T^\pi>t^\star]\le \Prb[T^{0}>t^\star]$, where $T^{0}$ is the completion time under full withholding of every cartel share.
\end{lemma}

\begin{proof}
Publicly committing cartel shares can only increase the public count $U_t$ pathwise. Hence $T^\pi\le T^0$ pathwise.
\end{proof}

For random agent sampling, the exact full withholding delay probability is $q^{(0)}=\Prb[S_{t^\star}>\Delta]$. The following concentration bound is often enough for parameter selection. Let
\(
    \theta=\frac{\Delta}{t^\star m}.
\) Since each hypergeometric draw is negatively associated and dominated in moment generating function by $\mathrm{Bin}(m,\beta)$, standard Chernoff optimization gives binomial KL exponents~\cite{hoeffding1963,chvatal1979,boucheron2013}.

\begin{proposition}[KL delay bounds]
\label{prop:kl}
For $\theta>\beta$,
\(
    q^{(0)}=\Prb[S_{t^\star}>\Delta]
    \le \exp\{-t^\star m\,\KL(\theta\Vert\beta)\}.
\)
For $\theta<\beta$,
\(
    1-q^{(0)}=\Prb[S_{t^\star}\le\Delta]
    \le \exp\{-t^\star m\,\KL(\theta\Vert\beta)\}.
\) The same statement holds for a stationary partial-withholding proxy after replacing $\Delta$ by $\Delta/(1-w)$.
\end{proposition}

The dependence on $\Delta$ produces a sawtooth as $\kappa$ varies. Within a fixed horizon $t^\star$, increasing $\kappa$ reduces slack and increases delay risk; when $\kappa$ crosses a multiple of $m$, a new slot opens and the slack resets. 
This yields a first design implication, avoid $m\mid\kappa$ unless the bounty is priced for a unilateral attack.

\section{Optimal Timely Completion Bounties}
\label{sec:pivot}
A completion bounty is sender funded. It is not a Groves or Clarke tax that elicits values~\cite{clarke1971,groves1973}, it is an external prize for a verifiable threshold event. The principal breaks budget balance because the benefit of timely completion accrues outside the team.

\begin{definition}[completion measurable bounty]
A budget $B$ completion bounty is a transfer rule $\varphi$ that maps committed shares to nonnegative payments for admissible shares. It is \emph{completion measurable} if payments depend only on the committed shares, the commitment schedule, and the validity of tickets or contribution proofs. It is \emph{timely} if all bounty payments are zero whenever $T>t^\star$. It is \emph{budget feasible} if total bounty paid is at most $B$ on every history. We exclude negative transfers and deposits in the main model.
\end{definition}

The optimal rule in this class is the uniform timely team bounty. The sender escrows $B$. If $T\le t^\star$, every admissible share committed no later than $t^\star$ receives $B/\Nstar$, where $\Nstar$ is the committed number of horizon opportunities. Any unpaid balance is returned to the sender. If $T>t^\star$, no bounty is paid and the escrow is returned. Fees are paid independently on admissible inclusions. In the coded dissemination instantiation, admissibility is enforced by agent bound tickets. A share occurrence is paid only if its ticket is valid, unspent, tied to the agent and slot, and publicly committed by the corresponding agent address. A copied symbol without the agent bound ticket is ignored for both decoding and bounty. The rule does not require estimating cartel size, delay value, agent quality, or transaction specific externalities. The timeliness condition is essential. A rule that pays after eventual completion lets a delaying coalition recover bounty in the recovery slot. \Team\ instead rewards on-time team output. It pays redundant horizon contributions as well as threshold crossing contributions because slack is exactly what makes sabotage coalitional, when $\Delta>0$, the attack deletes $\Delta+1$ shares, and every deleted horizon share should carry opportunity cost.

\begin{lemma}[Direct loss from deleting horizon shares]
\label{lem:loss}
Fix a contact path. Under \Team, if a coalition withholds $c$ admissible share opportunities before $t^\star$, then its pre-recovery direct revenue loss is at least
\(
    c f+\frac{c}{\Nstar}B .
\)
In particular, any minimal delay attack with $c=\Delta+1$ loses at least
\(
    (\Delta+1)f+\frac{\Delta+1}{\Nstar}B .
\)
\end{lemma}

\begin{proof}
Each withheld share forfeits its pre-horizon fee $f$. Under full submission, each horizon contribution also receives the bounty share $B/\Nstar$. Since the bounty expires on delay, a withheld horizon contribution cannot recapture its share in recovery. Counting only the withheld shares gives the stated lower bound; a coalition that also owns publicly committed shares at the horizon by agents may lose more when the bounty expires.
\end{proof}

\begin{lemma}[Symmetrization is without loss for worst case design]
\label{lem:symmetrization}
Consider the class of nonnegative completion-measurable timely completion bounties with budget $B$. Under an exchangeable availability process, for every contract $\varphi$ there is an anonymous contract $\bar\varphi$ whose worst case coalition deviation gain is weakly smaller.
\end{lemma}

\begin{proof}
Let $\Perm$ be the relabeling group on agents. For a relabeling $\sigma$, write $\sigma\varphi$ for the contract that first relabels the history by $\sigma$, applies $\varphi$, and then maps payments back. Define
\(
    \bar\varphi=\frac{1}{|\Perm|}\sum_{\sigma\in\Perm}\sigma\varphi .
\)
For any fixed contact path, coalition, and deviation, the coalition's transfer loss is affine in the contract. After maximizing over deviations, the deviation gain is a maximum of affine functions and hence convex in the contract. Therefore Jensen's inequality gives that the gain against $\bar\varphi$ is at most the average gain against the relabeled contracts. Exchangeability makes the worst case gain invariant under relabeling, so the averaged anonymous contract is weakly better in the minimax sense.
\end{proof}

\begin{theorem}[Exact minimax completion bounty]
\label{thm:minimax}
Fix a full submission horizon with $\Nstar=\kappa+\Delta$ admissible horizon shares. Among all nonnegative completion measurable timely completion bounties with budget $B$, the largest worst case bounty loss that can be guaranteed from any delay inducing deletion of $\Delta+1$ horizon shares is
\(
    \frac{\Delta+1}{\Nstar}B .
\)
The bound is attained by \Team.
\end{theorem}

\begin{proof}
For the upper bound, fix any budget-$B$ contract and any full submission on-time history. Let $p_1,\ldots,p_{\Nstar}$ be the bounty payments assigned to the $\Nstar$ horizon shares on that history. Since the contract is budget feasible, $\sum_i p_i\le B$. Hence the sum of the $\Delta+1$ smallest payments is at most $(\Delta+1)B/\Nstar$. If a coalition controls exactly those shares, withholding them deletes $\Delta+1$ contributions and delays completion by Theorem~\ref{thm:delay}, while forfeiting at most this much bounty. For the lower bound, \Team\ pays $B/\Nstar$ to every horizon share on the full submission path. Any delay inducing deletion contains at least $\Delta+1$ withheld horizon shares, so it forfeits at least $(\Delta+1)B/\Nstar$ in bounty. Thus no completion measurable nonnegative completion bounty can improve the worst case deletion guarantee.
\end{proof}
Concentrating the budget on the $\kappa$ pivotal shares raises the per-share reward but leaves the $\Delta$ redundant shares unpriced, and the cheapest sabotage set withholds exactly those. \Team\ instead prices every horizon share, so the minimal sabotage set forfeits $(\Delta+1)B/\Nstar > B/\kappa$ for $\kappa>1$. See Appendix~\ref{app:pivotk}.
\section{Coalition Stability, Pricing, and Ratchets}
\label{sec:coalitions}
Combining the slack geometry of Section~\ref{sec:geometry} with the bounty of Section~\ref{sec:pivot}, we establish the ex-post coalition condition under recovery accounting, then turn to deployment under slack uncertainty and exclusion ratchets.
\subsection{Ex-Post Coalition Stability}

On every path where a delay attack is possible, the coalition's minimum opportunity cost under \Team\ exceeds the value it can redistribute. Let a recovery slot be any slot that exceeds the honest horizon.

\begin{lemma}[Recovery slot farming bound]
\label{lem:recovery-farming}
Fix a delay inducing deviation with $c\ge \Delta+1$ pre-horizon withholdings and post-horizon deficit $d=c-\Delta>0$. Suppose the coalition's fee revenue from any single recovery slot is bounded by $R_1^+\le mf$, including same slot post-threshold overshoot inclusions. Then, under any recovery behavior by the cartel, $R_{\rm rec}(c) \le (d-1)f+R_1^+$.
\end{lemma}
\begin{proof}
Consider the first recovery slot where completion occurs. Before that slot, fewer than $d$ recovery shares have been publicly committed; otherwise the post-horizon deficit would already have been filled and completion would have occurred earlier. Thus the coalition's recovery fee revenue is at most $(d-1)f$ before completion, even on the worst path where every recovery inclusion belongs to the cartel before completion. The final recovery slot contributes at most $R_1^+$ by definition, including any same-slot overshoot shares that publicly commit after the threshold crossing share. Deliberately withholding in earlier recovery slots can only reduce the fee count before completion; it cannot increase the final-slot bound or the MEV cap $L$. Hence recovery slot farming cannot exceed the above.
\end{proof}

\begin{lemma}[Minimal delay is binding with recovery accounting]
\label{lem:binding}
Assume a deviation with $c\ge \Delta+1$ pre-horizon withholdings leaves deficit $d=c-\Delta$ and that one recovery slot contributes at most $R_1^+\le mf$ to the delaying coalition. Under \Team, for every $c\ge \Delta+1$, the coalition's net direct loss after crediting recovery fees is at least  \( c f+\frac{c}{\Nstar}B-R_{\rm rec}(c) \ge (\Delta+1)f+\frac{\Delta+1}{\Nstar}B-R_1^+ . \) Thus the minimal delay attack $c=\Delta+1$ is the binding deviation.
\end{lemma}

\begin{proof}
Let $d=c-\Delta$. By Lemma~\ref{lem:loss}, the pre-recovery direct loss is at least $(\Delta+d)f+(\Delta+d)B/\Nstar$. By Lemma~\ref{lem:recovery-farming}, recovery behavior can add at most $(d-1)f+R_1^+$ in incremental cartel fees. Therefore the net loss satisfies
\(
    N(d)
    \ge
    (\Delta+d)f+\frac{\Delta+d}{\Nstar}B-(d-1)f-R_1^+
    =
    (\Delta+1)f+\frac{\Delta+d}{\Nstar}B-R_1^+ .
\) This is minimized at $d=1$, equivalently $c=\Delta+1$.
\end{proof}

\begin{theorem}[Full submission is an ex-post strong equilibrium]
\label{thm:coalition}
Let $L=\alpha v\gamma^{t^\star}$ be an upper bound on the extractable value from delaying this transaction past the honest horizon. Suppose one recovery slot contributes at most $R_1^+\le mf$ incremental fee revenue to the delaying coalition. In the game $\Game(B)$ with \Team, if \( (\Delta+1)f+\frac{\Delta+1}{\Nstar}B \ge L+R_1^+, \) then the full submission strategy profile is ex-post strongly delay proof against every coalition, with arbitrary internal side payments bounded by the transaction's MEV and recovery revenue. Equivalently, it suffices to post
\begin{equation}
\label{eq:Bstar}
 B \ge \frac{\Nstar}{\Delta+1} \bigl(L+R_1^+-(\Delta+1)f\bigr)^+ .
\end{equation}
\end{theorem}

\begin{proof}
Any successful delay requires $c\ge \Delta+1$ withholdings by Theorem~\ref{thm:delay}. Lemma~\ref{lem:binding} shows that, after crediting all recovery fees available to a $c$-withholding deviation, the coalition's net direct loss is minimized at $c=\Delta+1$ and is at least
\(
    (\Delta+1)f+\frac{\Delta+1}{\Nstar}B-R_1^+ .
\) If the equation in the theorem holds, this net loss is at least $L$, so no side payment scheme funded by this transaction's extractable value and recovery revenue can make the deviation jointly profitable.
\end{proof}

\begin{corollary}[Positive slack blocks unilateral withholding]
If $\Delta>0$, no single withheld share can cause delay. A unilateral withholder loses fee $f$ and its bounty share $B/\Nstar$, but obtains no multi-slot delay option. Hence single contribution withholding is strictly dominated whenever $f+B/\Nstar>0$.
\end{corollary}

\begin{theorem}[Matching lower bound for completion measurable bounties]
\label{thm:lower}
Fix any nonnegative completion measurable timely completion bounty with budget $B$. Suppose the contact path contains a set of $\Delta+1$ horizon shares whose total bounty under full submission is at most $(\Delta+1)B/\Nstar$, and let the coalition control exactly those shares and no other horizon share. Suppose also that the induced recovery path gives the delaying coalition incremental post-horizon fee revenue $r_{\rm rec}$. If
\(
    (\Delta+1)f+\frac{\Delta+1}{\Nstar}B
    <
    L+r_{\rm rec},
\)
then there exists a delay inducing coalition deviation that is jointly profitable under side payments. Consequently, the budget in \eqref{eq:Bstar} is the exact worst-case budget scale for the full class of nonnegative completion measurable timely completion bounties.
\end{theorem}

\begin{proof}
By Theorem~\ref{thm:minimax}, every budget $B$ completion measurable bounty has some $\Delta+1$ horizon shares whose total bounty is at most $(\Delta+1)B/\Nstar$. A coalition controlling exactly those shares (and no other horizon share) withholds them. This deletes $\Delta+1$ contributions and delays completion by Theorem~\ref{thm:delay}. Because the coalition holds no submitted horizon share, the bounty it forfeits on delay is exactly the bounty on those $\Delta+1$ shares, at most $(\Delta+1)B/\Nstar$, together with the forfeited fees $(\Delta+1)f$ this bounds its direct pre-recovery loss above by the displayed quantity, and timeliness prevents bounty recapture after $t^\star$. If this loss is below $L+r_{\rm rec}$, the extracted value plus recovery fees can compensate all deviators with surplus.
\end{proof}

Restricting to a coalition that controls \emph{exactly} the cheap $\Delta+1$ shares is not a weakening of the bound but the reason it is tight. Because \Team\ expires the entire bounty on delay, a coalition that additionally controls any submitted horizon share forfeits that share's bounty as well, so its net cost strictly exceeds $(\Delta+1)f+\frac{\Delta+1}{\Nstar}B$. Larger coalitions are therefore strictly more deterred, and the worst case adversary is the minimal holding one of size exactly $\Delta+1$, this is the sense in which the cheapest sabotage set has size $\Delta+1$ and no larger.

\subsection{Bounty Pricing and Adaptive Ratchets}
\label{sec:pricing-dynamics}

We now leave the worst case pathwise certificate for deployment, how the bounty scales under Bayesian slack uncertainty and sender individual rationality, and how exclusion ratchets lower the required budget by imposing continuation costs on attackers.

\subsubsection{Sender pricing and individual rationality.}
The pathwise certificate costs \( B^{\rm path} = \frac{\Nstar}{\Delta+1} \bigl(L+R_1^+-(\Delta+1)f\bigr)^+ . \) If $R_1^+=mf$, then fees alone deter delay only when $(\Delta+1)f\ge L+mf$, which is impossible for positive $L$ because $\Delta+1\le m$. In the general availability model the relevant inequality is $\Delta+1\le a_{t^\star}$, the honest arrivals in the final horizon slot. Since $M_{t^\star-1}<\kappa$, we have $\Delta=M_{t^\star}-\kappa<M_{t^\star}-M_{t^\star-1}=a_{t^\star}$, the homogeneous model sets $a_{t^\star}=m$. Thus a pathwise certificate generally requires an external bounty unless recovery proposer surplus is closed by a timely fee prefix rule. With timely fee prefix, $R_1^+=0$ and fees alone deter delay whenever $(\Delta+1)f\ge L$. The \emph{timely fee prefix} is the fee variant that pays the direct fee $f$ to every admissible share publicly committed by the horizon $t^\star$, and only when completion is timely $(T\le t^\star)$; no share publicly committed after $t^\star$ is paid. The ``prefix'' is the temporal prefix of slots up to $t^\star$, not a rank prefix of the canonical order. Every horizon share is paid, so the fee, like the \Team\ bounty, prices the $\Delta$ redundant horizon shares and not only the $\kappa$ pivotal ones.

A sender need not always buy the pathwise certificate. When the cartel must fix its withholding before the slack is known, unnecessary deletions are punished on high slack states, and the resulting one shot Bayesian condition (Appendix~\ref{sec:pricing}) is typically far cheaper than the pathwise bound. A cruder posted price proxy sets the coalition's expected on-time bounty share $\beta B$ against expected delay value $(L+\E R_{\rm rec})q$, giving $B\approx ((L+\E R_{\rm rec})/\beta)q$. The sender buys the ex-post certificate only when it is individually rational, a conservative sufficient condition being $B^{\rm path}\le U_{\rm inc}-U_{\rm wh}$ (Appendix~\ref{sec:pricing}), low value transactions should use expected pricing or no bounty, while high value transactions may rationally buy the pathwise certificate.

\subsubsection{Adaptive Exclusion Ratchets.}
\label{sec:ratchet}

Agent bound tickets also give the sender an accountability signal. If a contacted agent fails to redeem its ticket, the sender can exclude that agent from future routing for the same transaction or epoch. This is a routing rule, not necessarily slashing: benign misses can be handled by buffers or probation. Its incentive effect is simple. If every withheld pre-horizon share causes continuation loss at least $\lambda$, then every delay inducing coalition pays an additional opportunity cost for each required sacrifice, and Theorem~\ref{thm:coalition} becomes
\( (\Delta+1)(f+\lambda)+\frac{\Delta+1}{\Nstar}B \ge L+R_1^+ . \) Equivalently,
\(
    B
    \ge
    \frac{\Nstar}{\Delta+1}
    \bigl(L+R_1^+-(\Delta+1)(f+\lambda)\bigr)^+ .
\) The proof is the same minimal-sabotage accounting as Lemma~\ref{lem:binding}, with $\lambda$ added per pre-horizon withholding. The ratchet also lowers future hit rates: after $F_t$ cartel agents are excluded, the eligible cartel fraction is
\(
    \beta_t=\frac{\beta n-F_t}{n-F_t}<\beta \quad(F_t>0).
\)
This probability effect is useful for expected pricing, but the robust statement is still the continuation cost inequality above. We distinguish two channels. Future exclusion imposes the per-withholding continuation cost $\lambda$ and is the basis of the robust inequality; it does not by itself prevent withholding accumulation across slots. A separate \emph{within horizon backfill} rule, in which the sender detects a no-show and contacts a replacement agent before the horizon, repairs early withholdings so that only a final slot burst can delay; this is the rule under which the single slot proxy $q_{\rm rat}=\Prb[A_{t^\star}>\Delta]$ of Section~\ref{sec:numerics} is valid.

\section{A Final Slot Visibility Impossibility}
\label{sec:micro}
The completion bounty prices multi-slot withholding. It cannot remove the final slot visibility floor. At the start of slot $t^\star$, the remaining honest deficit is
\(
    r=m-\Delta .
\) Let $V$ be the number of useful pre-seal contributions observed by a coalition in that slot, with $V=A_{t^\star}$ in the agent local coded dissemination model. Let $\rho(v,r)$ be the conditional probability that a coalition observing $v\ge r$ useful pieces can complete and act before the slot terminates.

\begin{proposition}[No completion measurable rule eliminates the final slot race]
\label{prop:micro-impossibility}
Assume that, on the event $V\ge r$, the coalition can take an extraction action before completion while still producing the same committed shares used by the bounty rule. Then no transfer rule measurable only with respect to the committed shares can reduce the physical race probability below
\(
    \Prb[\text{\emph{within-slot extraction}}]
    \ge
    \underline\rho(r)\Prb[V\ge r],
    \qquad
    \underline\rho(r)=\inf_{v\ge r}\rho(v,r).
\)
\end{proposition}

\begin{proof}
On $V\ge r$, the coalition has enough pre-seal information to attempt extraction before the committed shares are determined. By assumption, whether the coalition attempted extraction is not distinguished by the committed shares on which the transfer rule is measurable. Therefore the transfer vector is identical across the extraction and non-extraction histories. Any completion measurable bounty can change incentives for withholding committed shares, but it cannot condition on the pre-completion information event itself. The residual probability is therefore bounded below by the probability of the physical event times the least conditional success probability.
\end{proof}

In the agent local coded dissemination model,
\(
    \Prb[\text{within-slot extraction}]
    \le
    \bar\rho\,\Prb[A_{t^\star}\ge r],
    \qquad
    \bar\rho=\sup_{v\ge r}\rho(v,r).
\) If $r/m>\beta$, the hypergeometric Chernoff bound gives
\(
    \Prb[A_{t^\star}\ge r]
    \le
    \exp\{-m\KL(r/m\Vert\beta)\}.
\) On the knife edge $\Delta=0$, the exact feasibility probability is
\(
    \frac{\binom{\beta n}{m}}{\binom{n}{m}}
    \le
    \beta^m .
\) Thus completion measurable rewards solve the multi-slot sabotage problem, while threshold encryption, commit reveal timing, sealing rules, or stronger network assumptions are needed to reduce $V$ or $\rho$~\cite{f3b2022,choudhuri2024bte,bormet2025beatmev}.

\section{Related Work}
\label{sec:related}
The paper is closest in spirit to team production mechanism design. Holmstr\"om's impossibility shows the tension between efficient team effort and budget balance~\cite{holmstrom1982}. Our setting avoids the budget balance obstruction because the beneficiary of completion is outside the team and can fund an external prize. Unlike Groves or Clarke mechanisms~\cite{clarke1971,groves1973}, the rule does not elicit private values. It posts a fixed completion bounty for a verifiable threshold event. Delay is a coalitional deviation and the first $\Delta$ withholdings may be individually wasteful but jointly profitable once enough agents coordinate. Our solution concept is therefore strong equilibrium style, following the tradition of strong equilibrium and coalition proof Nash equilibrium~\cite{aumann1959,bernheim1987}. The ex-post condition asks whether any coalition, with arbitrary internal side payments funded by the delay value and recovery revenue, can profitably move the outcome from timely completion to late completion. The contract is deliberately anonymous and completion measurable. It does not require a model of agent types, cartel membership, latency distributions, or transaction specific externalities. This is the sense in which the construction follows Wilson's doctrine and the robust mechanism design perspective of Bergemann and Morris~\cite{wilson1987,bergemann2005}. The exact minimax result makes this robustness formal, once restricted to nonnegative completion measurable completion bounties, the uniform timely team bounty is worst case optimal. Blockchain fee market work studies how transaction fees, proposer incentives, and strategic block production interact~\cite{roughgarden2021fees,chungshi2023}. We use that line mainly as motivation for treating direct proposer revenue $f$ and recovery revenue $R_1^+$ explicitly. Rational-consensus baiting and slashing style protocols, such as TRAP~\cite{trap2022}, use penalties or future consequences. Our baseline result does not rely on deposits or negative transfers, the ratchet extension adds a continuation cost channel separately. MEV arises from order dependent execution and transparent pre-confirmation information~\cite{eskandari2019,daian2020,qin2021flashloans,babel2023lanturn}. High throughput BFT systems increasingly separate data dissemination from ordering and allow many validators to produce data concurrently, as in Narwhal/Tusk, DAG-Rider, and multiple concurrent-proposer designs~\cite{narwhal2022,dagrider2021,garimidi2025mcp}. The coding layer is rooted in information dispersal, rateless codes, AVID, and data availability sampling~\cite{rabin1989,luby2002,shokrollahi2006,cachin2005avid,alhaddad2022avid,albasam2018da,yu2019cmt,buterin2024danksharding}. Fair ordering and threshold encrypted mempools address related visibility and sequencing problems by different means~\cite{kelkar2020orderfairness,themis2023,speedyfair2024,quickorderfairness2022,f3b2022,choudhuri2024bte,bormet2025beatmev}. Our contribution is the incentive theory of threshold completion with slack.

\section{Conclusion}
Threshold completion with slack $\Delta$ creates a simple but sharp incentive problem. If $\Delta=0$, one withheld share delays completion. If $\Delta>0$, delay requires coordinated sabotage of at least $\Delta+1$ opportunities. A sender funded completion bounty can price this sabotage because the sender is an outside beneficiary who can break budget balance. The uniform timely team bounty \Team\ is optimal among nonnegative completion measurable completion bounties: with recovery accounting, the exact worst case budget scale is
\(
    B^\star
    =
    \frac{\Nstar}{\Delta+1}
    \bigl(L+R_1^+-(\Delta+1)f\bigr)^+ .
\)
Coded dissemination is one implementation of the model, alongside threshold signature collection, data-availability certification, and generic $k$-of-$n$ completion tasks. These results translate into a concrete design. The sender should provision strictly positive slack rather than set $m\mid\kappa$, bind every contribution to an agent identity or ticket so that admissibility is publicly verifiable, and pay all timely horizon contributions uniformly under \Team. Closing the recovery fee leak with a timely fee prefix sets $R_1^+=0$ and removes the bounty entirely in the fee-only region. The within-slot race of Section~\ref{sec:micro} falls outside this design and must be handled by sealing, commit-reveal, or threshold encryption rather than by the transfer rule.

\clearpage
\appendix
\renewcommand{\theHsection}{appendix.\Alph{section}}

\section{\Sedna\ Instantiation Details}
\label{app:sedna-specifics}

This appendix records the implementation details for the coded dissemination instantiation. They are not needed for the abstract threshold team production theorem, but they matter for deploying the mechanism in \Sedna. A \Sedna\ transaction consists of a public header, a transaction identifier $\id$, and an encoded payload. The sender ratelessly encodes the payload into a stream of symbols. In slot $t$, the sender contacts $m$ of the $n$ proposer agents and sends one fragment to each contacted agent. A fragment contains $s$ fresh symbol indices and symbols, together with an agent bound ticket. The payload decodes once $K$ distinct admissible symbols have been committed. In the fragment-level analysis we write
\[
    \kappa=\ceil{K/s}.
\]
When $K=\kappa s$, exactly $\kappa$ admissible shares suffice. Symbol level rounding $(\kappa=\lceil K/s\rceil)$ changes which shares are useful for decoding but not the bounty accounting, which is stated over the $N^\star$ horizon fragments; the delay condition remains $W_{t^\star}>\Delta$. Each agent $i$ has a reward address $\addr(i)$ for the relevant epoch. A fragment sent to agent $i$ in slot $t$ carries a ticket of the form
\[
    \ticket
    =
    \Sig_{\rm sender}
    \bigl(
        \id,t,i,\hash(\mathrm{body}),p,\addr(i),\mathrm{nonce}
    \bigr),
\]
where $p$ is the fee ticket and the nonce is unique. A fragment occurrence is admissible only if the ticket is valid, unspent, tied to agent $i$ and slot $t$ or to an explicit validity window, and signed by the corresponding agent address. Non-admissible copies do not count for decoding, direct fees, or bounty payments. This is the cryptographic plumbing that prevents bounty sniping, a proposer that copies another agent's symbols without a fresh agent bound ticket has not produced an admissible share. If duplicate symbol indices appear, \Sedna\ resolves them using a deterministic public order, for example lexicographic order by
\[
    (t,i,\hash(\ticket)).
\]
The important property is that the order is determined by committed shares and the sender's randomness, not by a proposer-controlled choice after seeing the transaction. Thus an agent's strategic action is inclusion or withholding of an admissible fragment, not manipulation of its rank within the bounty rule. The abstract team production game has agents, share opportunities, a completion threshold, and slack. In \Sedna:
\[
\begin{array}{rcl}
\text{agents} &\leftrightarrow& \text{proposer agents},\\
\text{share opportunity} &\leftrightarrow& \text{admissible agent-bound fragment},\\
\kappa &\leftrightarrow& \text{fragment threshold for decoding},\\
\Nstar &\leftrightarrow& t^\star m \text{ horizon fragment opportunities},\\
\Delta &\leftrightarrow& t^\star m-\kappa.
\end{array}
\]
Under full submission, the decode horizon is
\[
    t^\star=\ceil{\kappa/m},
    \qquad
    \Nstar=t^\star m,
    \qquad
    \Delta=\Nstar-\kappa.
\]
If a coalition withholds $W_{t^\star}$ admissible shares before $t^\star$, then the committed fragment count at the horizon is
\[
    \Nstar-W_{t^\star}
    =
    \kappa+\Delta-W_{t^\star}.
\]
Therefore public decoding is delayed past the honest horizon if and only if
\[
    W_{t^\star}>\Delta.
\]
This is exactly Theorem~\ref{thm:delay} specialized to coded dissemination. The ticket format, symbol resolution order, and fragment level rounding are \Sedna-specific. The incentive theorem is not. Once admissible shares are verifiable and the horizon has $\Nstar=\kappa+\Delta$ share opportunities, the optimal timely team bounty and its lower bound depend only on $(\kappa,\Delta,f,L,R_1^+)$, not on the fact that the underlying object is a coded transaction payload.

\section{A Sedna Native Prefix Bounty Variant}
\label{app:pivotk}
The main text uses the uniform timely team bounty \Team, which pays all admissible horizon contributions when completion is timely. This appendix records the \Sedna~native prefix bounty rule \PivotK. The rule is useful as an implementation benchmark and as a clean rank based mechanism, but it should not be confused with the globally optimal completion-measurable bounty in Theorem~\ref{thm:minimax}.

\paragraph{Definition of \PivotK.}
Suppose the transaction decodes by the honest horizon $t^\star$. The protocol scans committed admissible fragment occurrences in the deterministic symbol resolution order and identifies the first $K$ distinct symbol indices used for decoding. Let $I^\star$ denote this pivotal index set. For each $j\in I^\star$, let $\mathrm{agent}(j)$ be the agent whose admissible fragment first contributed index $j$. The \PivotK\ rule pays
\[
    \frac{B}{K}
\]
to $\mathrm{agent}(j)$ for each $j\in I^\star$, and pays no bounty to non-pivotal symbols. If $T>t^\star$, the bounty expires and no bounty is paid. Under the exact divisibility convention $K=\kappa s$, each pivotal fragment contributes $s$ pivotal symbols. Hence \PivotK\ is equivalently the fragment level rule that pays
\[
    \frac{B}{\kappa}
\]
to each of the first $\kappa$ admissible shares in the canonical order, provided decoding occurs by $t^\star$.

\paragraph{Direct loss under \PivotK.}
At the honest horizon there are $\kappa+\Delta$ admissible fragments. Under \PivotK, only $\kappa$ of them are paid the bounty, the remaining $\Delta$ are redundant for both decoding and the prefix bounty. Therefore, if a coalition withholds $c$ horizon fragments, its pre-recovery direct loss is at least
\[
    c f+\frac{(c-\Delta)^+}{\kappa}B .
\]
For a minimal delay attack, $c=\Delta+1$, this becomes
\[
    (\Delta+1)f+\frac{B}{\kappa}.
\]

\paragraph{Within-class minimax property.}
A prefix linear pivotal bounty chooses nonnegative weights
\[
    \omega_1,\ldots,\omega_\kappa,
    \qquad
    \sum_{j=1}^{\kappa}\omega_j=1,
\]
and pays $\omega_jB$ to the $j$-th pivotal fragment. Let $\omega_{(1)}\le\cdots\le\omega_{(\kappa)}$ be the nondecreasing rearrangement and define
\[
    S_d(\omega)=\sum_{i=1}^{d}\omega_{(i)}.
\]
Then for every $d\in\{1,\ldots,\kappa\}$,
\[
    S_d(\omega)\le \frac{d}{\kappa}.
\]
The uniform rule $\omega_j=1/\kappa$ attains equality. Thus \PivotK\ maximizes the worst case bounty forfeiture from deleting any fixed number $d$ of pivotal ranks among anonymous prefix linear pivotal rules.

\begin{proof}
Let $a=S_d(\omega)/d$ be the average of the $d$ smallest weights. Every remaining weight is at least $a$, so
\[
    1=\sum_{j=1}^{\kappa}\omega_j
    \ge
    \kappa a.
\]
Therefore $S_d(\omega)=da\le d/\kappa$. Equality forces all weights to equal $1/\kappa$.
\end{proof}

\paragraph{Comparison with the timely team bounty.}
The distinction is important. \PivotK\ is minimax only after the designer commits to paying the pivotal prefix. It leaves the $\Delta$ redundant horizon fragments unpaid. A worst case minimal sabotage set can therefore delete $\Delta$ redundant fragments and one pivotal fragment, burning only one bounty share.

By contrast, \Team\ pays every admissible horizon contribution $B/\Nstar$ when completion is timely. A minimal sabotage set of size $\Delta+1$ therefore burns
\[
    \frac{\Delta+1}{\Nstar}B
\]
in bounty. The corresponding pathwise budget scales are
\[
    B_{\PivotK}^{\star}
    =
    \kappa\bigl(L+R_1^+-(\Delta+1)f\bigr)^+,
\]
whereas
\[
    B_{\Team}^{\star}
    =
    \frac{\Nstar}{\Delta+1}
    \bigl(L+R_1^+-(\Delta+1)f\bigr)^+.
\]
For $\Delta=0$, the two coincide because $\Nstar=\kappa$. For every positive slack instance with $\kappa>1$,
\[
    \frac{\Nstar}{\Delta+1}
    =
    \frac{\kappa+\Delta}{\Delta+1}
    <
    \kappa,
\]
so \Team\ gives a strictly smaller worst case budget certificate.

\section{Parameter Choice and Bounty Pricing}
\label{sec:pricing}

The preceding sections give a pathwise stability certificate and the main text already separates the within-slot residual. This appendix turns the multi-slot result into sender guidance. There are two pricing quantities, the \emph{ex-post certificate} needed to rule out every delay inducing coalition on every possible contact path, and the \emph{expected price} a risk neutral sender might post when attacks occur only on rare paths.

\subsection{Slack as an incentive lever}

Fix a target latency horizon $h\ge 1$ and a contact rate $m$. All thresholds $\kappa$ with $t^\star=h$ can be written as
\[
    \kappa=hm-\Delta,
    \qquad \Delta\in\{0,\ldots,m-1\}.
\]
For this fixed horizon, write $L_h=\alpha v\gamma^h$ and let $R_{1,h}^+$ be the single recovery slot fee upper bound used for this latency target. The ex-post pivotal bounty threshold from Theorem~\ref{thm:coalition} is
\begin{equation}
\label{eq:path-bounty-delta}
    B^{\rm path}_h(\Delta)
    =
    \frac{hm}{\Delta+1}
    \bigl(L_h+R_{1,h}^+-(\Delta+1)f\bigr)^+ .
\end{equation}

\begin{theorem}[Slack dominance for multi-slot incentive cost]
\label{thm:slack-dominance}
For fixed $h,m,L_h,R_{1,h}^+$ and $f$, the function $B^{\rm path}_h(\Delta)$ is weakly decreasing in $\Delta$ on $\{0,\ldots,m-1\}$. Thus, subject to a fixed latency horizon and a fixed recovery fee bound, larger slack weakly reduces the ex-post bounty needed to deter multi-slot withholding. The zero slack knife edge is the most expensive point.
\end{theorem}
\begin{proof}
Let $A=L_h+R_{1,h}^+$. On the active region,
\(
    B^{\rm path}_h(\Delta)=hm\bigl(\tfrac{A}{\Delta+1}-f\bigr),
\)
which is weakly decreasing in $\Delta$. Once the positive part becomes inactive the function is zero and remains zero for larger $\Delta$. Hence the pathwise bounty threshold is weakly decreasing in slack.\end{proof}

Theorem~\ref{thm:slack-dominance} gives a formal basis for preferring strictly positive slack. The worst case bounty $B^{\rm path}_h(\Delta)$ is weakly decreasing in $\Delta$, so the zero slack configuration $m\mid\kappa$ is the most expensive point and should be selected only when the bounty is already priced for a unilateral attack. However, slack is not free. The final slot deficit is
\[
    r=m-\Delta,
\]
so the agent local within-slot feasibility probability is
\[
    q_{\rm micro}(\Delta)=\Prb[A_{t^\star}\ge m-\Delta].
\]
This tail is weakly increasing in $\Delta$, because the required number of cartel contacts falls as slack rises. Thus slack moves the two risks in opposite directions: it makes multi-slot withholding cheaper to deter, but it makes pre-seal decoding easier if timing is loose. This is a useful separation. \Pivot\ and the ratchet price the multi-slot sabotage game; sealing deadlines, commit-reveal, or encryption must handle the within-slot visibility floor.

The recovery-fee term changes the fee only message. In the default pathwise fee-leak model $R_{1,h}^+=mf$, so $(\Delta+1)f\ge L_h+mf$ is impossible for any positive $L_h$ because $\Delta+1\le m$. Thus, with ordinary recovery fees left open, the pivotal bounty is always load bearing in a pathwise certificate. The familiar fee only region returns only if the fee leak is closed: under the timely fee prefix variant $R_{1,h}^+=0$, so no bounty is needed whenever $(\Delta+1)f\ge L_h$. For expected pricing, one may instead use $\E R_1^+\approx\beta mf$, yielding a non-empty fee only region when $(\Delta+1-\beta m)f\ge L_h$.

\subsection{Bayesian slack uncertainty}

The ex-post analysis grants the cartel exact knowledge of $\Delta$. That is conservative. In coded dissemination, $\kappa$ can depend on the hidden payload size and the sender's symbolization choice, so a cartel may need to choose its first withholding action before knowing how much slack exists. The next proposition gives the corresponding one-shot Bayesian condition.

Let $\mu$ be the cartel's prior over possible pairs $(\Delta,\kappa)$ at the moment it chooses how many addressed share opportunities to withhold. Consider a one-shot deletion plan that withholds $c$ opportunities before learning the realized slack. This abstraction is pessimistic in two ways: it assumes the cartel has enough contacts to implement the chosen $c$, and it ignores ratchet continuation costs.

\begin{proposition}[Bayesian no withholding condition under slack uncertainty]
\label{prop:bayes-slack}
Under timely \Pivot, the expected gain from a one-shot plan that withholds $c\ge1$ share opportunities is at most
\[
    (L+R_1^+)\,\Pr_\mu[\Delta<c]
    +f\,\E_\mu\!\left[((c-\Delta)^+-1)^+\right]
    -cf
    -B\,\E_\mu\!\left[\frac{c}{\Nstar}\right].
\]
Therefore full submission is Bayesian incentive compatible against all such one-shot deletion plans if, for every $c\ge1$,
\begin{equation}
\label{eq:bayes-slack-ic}
    cf+B\,\E_\mu\!\left[\frac{c}{\Nstar}\right]
    \ge
    (L+R_1^+)\Pr_\mu[\Delta<c]
    +f\,\E_\mu\!\left[((c-\Delta)^+-1)^+\right].
\end{equation}
\end{proposition}
\begin{proof}
Given realized $\Delta$, withholding $c$ opportunities delays completion only when $d=(c-\Delta)^+>0$, giving MEV at most $L\mathbf 1\{\Delta<c\}$. By Lemma~\ref{lem:recovery-farming}, recovery behavior contributes at most $R_1^+\mathbf 1\{\Delta<c\}+f(d-1)^+$ in incremental cartel fees. The plan always forfeits $cf$ in direct fees. Under \Team, each withheld horizon contribution also forfeits its full submission bounty share $B/\Nstar$, whether or not the attempted deletion is sufficient to cause delay. Taking expectations over $\mu$ gives the displayed upper bound.
\end{proof}

Condition~\eqref{eq:bayes-slack-ic} can be much cheaper than the ex-post bound because unnecessary deletions are punished on states with high slack. In particular, if the prior puts significant probability on $\Delta\ge c$, the cartel pays the fee cost $cf$ while receiving no delay benefit. This is the formal version of the intuition that hidden dissemination parameters help the sender: the cartel must either under-withhold and fail on low slack states or over-withhold and burn revenue on high slack states.

\paragraph{Posted bounty and sender individual rationality.}
If the withholding proof threshold is a random variable $B^*(\omega)$ with CDF $F$, the sender's expected utility from posting $B$ is $F(B)(U_{\rm inc}-B)+(1-F(B))U_{\rm wh}$. This gives a standard posted price tradeoff: routine transactions may use expected pricing, while high value transactions should pay for the ex-post certificate when it is individually rational. 

\section{Illustrative Calibration}
\label{sec:numerics}
For $n=100$, $m=20$, and $\beta=0.2$, exact hypergeometric tails show the sawtooth predicted by the theory. At the knife edge $\kappa=20$ ($t^\star=1$, $\Delta=0$), full withholding delays with probability $0.993$; at $\kappa=30$ ($t^\star=2$, $\Delta=10$), the static delay probability is $0.136$, while the one shot ratchet proxy is $8.0\times10^{-5}$. The expected pricing proxy from Section~\ref{sec:coalitions} is $B_{\rm exp}\approx ((L+\E R_{\rm rec})/\beta)q$. For the back-of-envelope table we take $\gamma\approx1$ and drop the recovery credit, giving the simplified $(\alpha v/\beta)q$ rule; with protected value-at-risk $\alpha v=100$ and $\beta=0.2$, this falls from about $68$ under static routing to about $0.04$ under the ratchet proxy.

This is a heuristic for ordinary pricing, not the pathwise certificate in Theorem~\ref{thm:coalition}. The exact hypergeometric probabilities are below. Here $q^{(0)}=\Prb[S_{t^\star}>\Delta]$ is the full withholding multi-slot delay probability. The ratchet proxy $q_{\rm rat}=\Prb[A_{t^\star}>\Delta]=\Prb[\mathrm{Hypergeom}(n,\beta n,m)>\Delta]$ assumes the within-horizon backfill rule of Section~\ref{sec:ratchet}, under which only a final-slot burst exceeding the slack can delay completion. Since slots are i.i.d., $A_{t^\star}\overset{d}{=}A_1$, so the proxy depends only on $\Delta$, not on $t^\star$. This is a deployment heuristic, not a pathwise certificate; the robust statement remains the continuation cost inequality of Section~\ref{sec:ratchet}. Note that pure future exclusion does \emph{not} give single slot behavior, since cartel members not yet selected remain eligible and can withhold in later slots; under exclusion alone the honest delay probability stays near $q^{(0)}$. Finally $q_{\rm micro}=\Prb[A_{t^\star}\ge r]$ is the within-slot feasibility probability under $V=A_{t^\star}$ and $\bar\rho=1$.
\begin{center}
\begin{tabular}{@{}rrrrccc@{}}
\toprule
$\kappa$ & $t^\star$ & $\Delta$ & $r$ & $q^{(0)}$ & $q_{\rm rat}$ & $q_{\rm micro}$ \\
\midrule
10  & 1 & 10 & 10 & $8.0\!\times\!10^{-5}$ & $8.0\!\times\!10^{-5}$ & $6.5\!\times\!10^{-4}$ \\
20  & 1 & 0  & 20 & $0.993$ & $0.993$ & $1.9\!\times\!10^{-21}$ \\
30  & 2 & 10 & 10 & $0.136$ & $8.0\!\times\!10^{-5}$ & $6.5\!\times\!10^{-4}$ \\
50  & 3 & 10 & 10 & $0.699$ & $8.0\!\times\!10^{-5}$ & $6.5\!\times\!10^{-4}$ \\
100 & 5 & 0  & 20 & $\approx1$ & $0.993$ & $1.9\!\times\!10^{-21}$ \\
\bottomrule
\end{tabular}
\end{center}

\section{Proof Details for the KL Bounds}
\label{app:kl}
Let $A\sim\mathrm{Hypergeom}(n,\beta n,m)$. Sampling without replacement is more concentrated than sampling with replacement, so for $\lambda>0$,
\[
    \E[e^{\lambda A}]\le (1-\beta+\beta e^\lambda)^m .
\]
For $S_t=\sum_{u\le t}A_u$, independence across slots gives
\[
    \E[e^{\lambda S_t}]
    \le (1-\beta+\beta e^\lambda)^{tm}.
\]
Chernoff's inequality and optimization over $\lambda$ yield
\[
    \Prb[S_t\ge \theta tm]\le \exp\{-tm\KL(\theta\Vert\beta)\},
    \qquad \theta>\beta.
\]
The lower tail bound follows analogously with $\lambda<0$.


\section{Stationary Partial Withholding Proxy}
\label{app:stationary}
A simpler proxy lets the cartel publicly commit a fraction $w\in[0,1]$ of its addressed shares. The committed count is
\[
    U_t(w)=tm-(1-w)S_t.
\]
At $t^\star$, delay occurs iff $(1-w)S_{t^\star}>\Delta$. Defining
\[
    \theta_w=\frac{\Delta}{(1-w)t^\star m},
\]
one obtains the same KL upper or lower tail bounds as Proposition~\ref{prop:kl} whenever $\theta_w\in(0,1)$. The proxy is useful for numerical pricing but is not the equilibrium model, dynamic minimal sabotage is cheaper than stationary withholding.

\section{Minimal Sabotage and Net Marginal Fees}
\label{app:negative-fees}
Theorem~\ref{thm:coalition} assumes that publicly committing an additional admissible pre-horizon share has nonnegative net marginal payoff before considering MEV, and it explicitly accounts for post-horizon recovery fees through $R_1^+$ and Lemma~\ref{lem:recovery-farming}. The nonnegative pre-horizon condition holds if the share carries a reserved fee ticket, if the agent has reserved bandwidth for coded shares, or if the ticket price covers any displacement cost. If publicly committing a coded share displaces a more valuable alternative transaction, minimal sabotage can fail: conditional on already causing delay, a cartel may prefer to withhold additional shares to free capacity. The general condition needed for the main theorem is therefore that every extra pre-horizon admissible inclusion weakly increases direct payoff and weakly increases the cartel's timely pivotal count. Under \Team, the second part is pathwise because every additional admissible pre-horizon inclusion by the cartel weakly increases the cartel's paid horizon contribution count whenever completion is timely.

\end{document}